\documentclass[journal]{IEEEtran}

\usepackage{amsmath,amssymb,amsfonts,amsthm}
\usepackage{color}
\usepackage{graphicx}
\usepackage{algpseudocode} 
\usepackage{algorithm}

\newtheorem{lemma}{\textbf{Lemma}}
\newtheorem{cor}{\textbf{Corollary}}
\newtheorem{definition}{\textbf{Definition}}
\newtheorem{thm}{\textbf{Theorem}}

\newtheorem{proposition}{\textbf{Proposition}}

\newtheorem{remark}{\textbf{Remark}}

\newcommand{\ie}{{i.e.}}
\newcommand{\eg}{{e.g.}}

\newcommand{\mc}[1]{\mathcal{#1}}

\newcommand{\mbb}[1]{\mathbb{#1}}

\renewcommand{\Pr}[1]{\ensuremath{\operatorname{\mathbf{Pr}}\left[#1\right]}}
\newcommand{\Ex}[1]{\ensuremath{\operatorname{\mathbf{E}}\left[#1\right]}}

\newcommand{\pt}{\tilde{p}}

\DeclareMathOperator{\dist}{dist}

\title{Coded Load Balancing in Cache Networks}

\author{Mahdi~Jafari~Siavoshani, Farzad~Parvaresh, Ali~Pourmiri, Seyed~Pooya~Shariatpanahi%
\thanks{The authors' names appear in alphabetical order.}%
\thanks{
	M. Jafari Siavoshani is with the Department of Computer Engineering, Sharif University of Technology, Tehran, Iran (email: mjafari@sharif.edu).}%
\thanks{
	F. Parvaresh is with the Department of Electrical Engineering, University of Isfahan, Isfahan, Iran (email: f.parvaresh@eng.ui.ac.ir) and also School of Mathematics,
	Institute for Research in Fundamental Sciences (IPM), P.O. Box: 1395--5746, Tehran,
	Iran. F. Parvaresh was supported by a grant from IPM (No. 95680425).}%
\thanks{
	A. Pourmiri is with the Department of Software Engineering, University
		of Isfahan, Isfahan, Iran, and Department of Computing, Macquarie University, Sydney, Australia (ali.pourmiri@mq.edu.au).}%
\thanks{
	S. P. Shariatpanahi is with the School of Computer Science, Institute for Research in Fundamental Sciences (IPM), Tehran, Iran (email: pooya@ipm.ir).}
}

\begin{document}

\maketitle

\begin{abstract}
We consider load balancing problem in a cache network consisting of storage-enabled servers forming a distributed content delivery scenario. Previously proposed load balancing solutions cannot perfectly balance out requests among servers, which is a critical issue in practical networks. Therefore, in this paper, we investigate a coded cache content placement where coded chunks of original files are stored in servers based on the files popularity distribution. In our scheme, upon each request arrival at the delivery phase, by dispatching enough coded chunks to the request origin from the nearest servers, the requested file can be decoded.

Here, we show that if $n$ requests arrive randomly at $n$ servers, the proposed scheme results in the maximum load of $O(1)$ in the network. This result is shown to be valid under various assumptions for the underlying network topology. Our results should be compared to the maximum load of two baseline schemes, namely, \emph{nearest replica} and \emph{power of two choices} strategies, which are $\Theta(\log n)$ and $\Theta(\log \log n)$, respectively. This finding shows that using coding, results in a considerable load balancing performance improvement, without compromising communications cost performance. This is confirmed by performing extensive simulation results, in non-asymptotic regimes as well.
\end{abstract}


\centerline{\textbf{Keywords}}
Distributed Caching Servers, Content Delivery Networks, Coded Caching, Request Routing, Load Balancing, Communication Cost.

\section{Introduction}

Here, in the first subsection we present the motivation and the main problem we consider in this paper. Then, in the second subsection, we review related works which consider using coding for delivering contents to the users, and discuss how our work differs from them. Finally, we present the paper structure.

\subsection{Motivation and Paper Contributions}
The main objective of a Content Delivery Network (CDN) is fulfilling end-users' content requests by forwarding these requests to distributed caching servers. Such forwarding procedure (aka. request routing) from the request origin to CDN servers is done by a \emph{mapping scheme} that decides which request should be answered by which server, considering the current network state \cite{SuCKB09, ChenST15}. However, due to the random nature of these requests, addressing the load balancing issue among the servers is of critical importance \cite{ManfrediOR13, Cardellini_2002_Survay}, i.e., overloading a single server with many requests should be avoided. Moreover, to address scalability issues, the main research in this field has been focused on designing effective \emph{distributed} load balancing schemes \cite{Dilly_2002_Distributed_LB, Nygren_Akamai_2010}.

Every request-to-server mapping scheme should manage two main metrics when assigning requests to caching servers. The first metric is \emph{communication cost} which is a measure of distance between the assigned server and the request origin, while the second metric is the \emph{maximum load} incurred to servers, which is the number of requests assigned to the most loaded server. Ideally, one aims to design a mapping scheme that results in the minimum communication cost, while evenly distributing the requests among servers. However, due to the random nature of request arrivals this is not always possible. Interestingly, there is an intrinsic trade-off between these two metrics, and managing this fundamental trade-off is a core issue in load balancing schemes \cite{Pathan_OTM_2008, TangTW14, PourJafShar_BalAllocCDN_IPDPS,SiaPourShar_TPDS17}. While assigning each request to the nearest eligible server\footnote{Here, by eligible servers we mean those servers that have cached the requested file.} (\ie, the {nearest replica strategy}), results in the minimum communication cost, it may incur high loads to servers in proximity of request flash crowds. This observation has resulted in proposing a \emph{proximity-aware power of two choices} strategy in  \cite{PourJafShar_BalAllocCDN_IPDPS} and \cite{SiaPourShar_TPDS17}, which queries the load of two nearby servers and assigns the request to the server with lesser load. As proved in \cite{SiaPourShar_TPDS17}, this will reduce the maximum load of $\Theta(\log n)$ in the nearest replica strategy to $\Theta(\log \log n)$, at a certain increase in communication cost.

In contrast to all previous load balancing schemes in CDNs, in this paper we follow a fundamentally different approach rather than allocating each request to a \emph{single} responding server. We propose to cache coded chunks of each file in distributed caching servers, which either consist of random linear combinations of the original file chunks, or are constructed via more sophisticated Fountain-like codes \cite{Luby_2002_LT_Codes,Byers_1998_Fountain_Code, Shokrollahi06_Raptor_Code}. In the delivery phase, each request is distributed among a number of nearby servers which have cached the corresponding coded chunks. Then, the requester can decode the whole file if it receives enough coded chunks from such servers.

In particular, we model the cache network topology by an underlying graph of $n$ nodes which represent cache-enabled servers, where the edges are communication links. At the cache content placement phase, each file is divided into $\ell$ equal-sized chunks, which are then linearly combined to form coded chunks corresponding to this particular file. Then, each server caches coded chunks corresponding to different files, based on the file popularity profile. At the delivery phase, $n$ file requests arrive uniformly at random at different servers, where files are requested based on the popularity profile. Then, in order to satisfy each request, $\ell$ nearest coded chunks corresponding to that file are routed to the request origin via the shortest path on the graph. Assuming independent linear combinations in different coded chunks, it is easy to see that each file can be decoded at the request origin server. 

In this setup, we consider Zipf file popularity distribution, and network topologies such as grid, random geometric graphs, random $d$-regular graphs, and hypercube. 
We prove that our proposed coded scheme will
result in the maximum load of $O(1)$, if each file is divided into $\ell=\Theta(\log n)$ chunks.
This result should be compared the the maximum load of nearest replica and power of two choices strategies which are $\Theta(\log n)$ and $\Theta(\log \log n)$, respectively \cite{SiaPourShar_TPDS17}.
Interestingly, for grid networks, we show that our proposed coded scheme will result in asymptotically the same communication cost as the baseline uncoded nearest replica strategy. Furthermore, we investigate our findings in finite size networks by performing extensive simulations, through which we also confirm the superiority of our proposal compared to the previous proposed schemes.

\subsection{Related Works}
Using coding in content delivery scenarios has been proposed in previous works. The authors in \cite{Gkantsidis_Rodriguez_Infocom05_NetCode_CDN} and \cite{Lee_2006_CodeTorrent} consider the benefit of using network coding in a P2P based  scenario and VANETs, respectively. Also, the papers \cite{Bilal19} and \cite{Kumar18} investigate the role of network coding in Information Centric Networks, and Critical Infrastructure Networks, respectively. Moreover, \cite{Guven_MultiMulticast_2008} considers the effect of coding on the multiple multicast problem. Interested reader can also see \cite{Magli_2013_NetCodeCDN_aReview} which is a good review on the role of network coding for multimedia delivery. 

On another line of research, the in-network caching idea has been proposed to relieve network congestion as explained in the following. The authors in \cite{Shanmugam_2013_FemtoCaching} investigate optimal coded/uncoded cache content placement for minimizing delivery delay in wireless content delivery scenarios. Following the results in \cite{Shanmugam_2013_FemtoCaching} many researchers have considered the role of caching in wireless settings such as \cite{Yin19}, \cite{Teng19}, and \cite{Lee19}.

Also, in \cite{BGL15}, the authors look at the problem of optimal MDS coded cache placement in wireless edge networks, which is extended to heterogeneous file and cache sizes in \cite{Liao17}. Furthermore the authors in \cite{Zhang18} consider MDS coded cache content placement from an energy consumption perspective. Finally, as it is shown in \cite{Recayte18}, LT codes can be used at the wireless edge caches to minimize backhaul rate.

In all above works, metrics such as the communication rate/cost, delay, and energy consumption in content delivery networks are investigated, and the benefits of using codes are discussed from different perspectives. However, an important aspect which is ignored in previous works is the load balancing issue. In contrast to all the aforementioned research works in this context, our paper is the first  work which investigates the role of coding in load balancing performance through an analytical approach. Here, we derive closed-form asymptotic results for the impact of coding on the communication cost and maximum load of servers in cache networks. Through the framework developed here, we investigate communication cost and maximum load in grid networks, which is then generalized to other network topologies. Moreover, we confirm our analytical findings through extensive simulations, via which we also investigate the effect of network parameters such as number of servers, cache size, popularity profile, and network topology.  

\subsection{Paper Structure}
The rest of paper is organized as follows. Our system model and performance metrics are introduced in Section~\ref{sec:SystemModel}. Then, the proposed coded scheme is presented in Section~\ref{sec:CodedLoadBalancingGrid} and its maximum load and communication cost are asymptotically analyzed for grid networks. Next, our results are extended to more general networks in Section~\ref{sec:CodedLoadBalancingGeneral}. In Section~\ref{sec:Simulations}, simulation results for finite-sized networks are presented. Finally, the paper is concluded in Section~\ref{sec:Conclusions}.

\section{System Model}\label{sec:SystemModel}

\subsection{Notation}
Throughout the paper, \emph{with high probability} (w.h.p.) refers to an event that happens with probability $1-1/n^c$, as $n \rightarrow \infty$, for some constant $c>0$. 
Let $G=(V,E)$ be a graph with vertex set $V$ and edge set $E$. 
For $u\in V$ let $\deg(u)$ denote the degree of $u$ in $G$. For every pair of nodes $u,v\in V$,   $\dist_G(u, v)$ denotes the length (number of edges along the path) of a shortest path from node $u$ to $v$ in graph $G$. The neighborhood of $u$ at distance $r$ in  $G$ is defined as    
\[
  B_r(u) \triangleq \left\{v : \dist_G(u, v)\le r ~ {\text {and}}~~ v\in V(G) \right\}.
\]
For a set $A$ we use $\bar{A}$ to denote its complement. To show the complement of an event $\mc{E}$ we use $\neg \mc{E}$.
We use $\mathrm{Po}(\lambda)$ to denote the Poisson distribution with parameter $\lambda$.
The expected value of a random variable $X$ is denoted by $\Ex{X}$.
%
%
The operator ``$\circ$'' represents the concatenation operator, i.e., the file $W = W_1 \circ W_2$ is generated by appending
the file $W_2$ at the end of the file $W_1$.
Throughout the paper, $H(\cdot)$ represents the binary entropy (\ie, information content), measured in bits.

For asymptotic notation, we use $g(n)=O(f(n))$ if there exist $c$ and $n_0$ such that for $n>n_0$ we have $g(n)< c f(n)$. In this case we also write $f(n)=\Omega(g(n))$. Moreover, we write $g(n)=o(f(n))$ if $\lim_{n\rightarrow\infty} g(n)/f(n)=0$. In this case one can also write $f(n)=\omega(g(n))$. Finally, if $f(n) = O(g(n))$ and $f(n) = \Omega(g(n))$ then we write $f(n) = \Theta(g(n))$.

\subsection{Problem Setting}
We consider a cache network consisting of $n$ caching servers (also called cache-enabled servers) and undirected edges connecting neighboring servers forming a graph $G$. Direct communication is possible only between adjacent servers, and other communications should be carried out in a multi-hop fashion through the network. 

Suppose that the cache network is responsible for handling a library of $K$ files $\mathcal{W}=\{W_1,\dots,W_K\}$, each of size $F$ bits, whereas the popularity profile follows a known distribution $\mathcal{P}=\{p_1,\dots,p_K\}$.

We assume that the network operates in two phases, namely, \emph{cache content placement} and \emph{content delivery}. In the cache content placement phase, each server $i$ caches $Z_i \triangleq \Psi_i(W_1,\ldots,W_K)$ such that $H(Z_i)\le MF$ where $M$ is the cache size of each server. 
Here $\Psi_i(\cdot)$ is a function of files in $\mathcal{W}$ that generates data $Z_i$ to be placed in the cache of server $i$.

Consider a time block during which $n$ file requests have arrived uniformly at random among the servers (\ie, graph vertices). Let $R_i$ denote the number of requests (demands) arrived at server $i$. Then, for large $n$, we have $R_i \sim \mathrm{Po(1)}$ for all $1\le i\le n$.

For the library popularity profile $\mc{P}$, we consider the Zipf distribution with parameter $\gamma \geq 0$, where the request probability of the $k$-th popular file is inversely proportional to its rank as follows
\begin{equation}\label{eq:ZipfDist}
p_k=\frac{1/k^\gamma}{\sum\limits_{l=1}^{K}1/l^\gamma},\quad k=1,\dots,K,
\end{equation}
which has been confirmed to be the case in many practical applications \cite{Zipf1_99,Zipf2_07}.

In the content delivery phase, suppose server $i$ has received a set of requests $\mc{W}_i \triangleq \{ W_{f_{i,1}},\ldots, W_{f_{i,R_i}}  \}$, where $f_{i,j}$ is the file index of $j$th  received request of server $i$.
In order to satisfy the demand $W_{f_{i,j}}$, a set of messages $M_{s\rightarrow i} ^{(j)}$ for $s\in\{1,\ldots,n\}$ will be sent from other servers to server $i$ where $M_{s\rightarrow i} ^{(j)}$ (which is a function of $Z_s$) is the message sent from server $s$ to server $i$ to satisfy $j$th request of server $i$.
Then, we say that server $i$ can \emph{successfully decode} the request $W_{f_{i,j}}$ if there exists a decoding function $\Phi_{i,j}(\cdot)$ such that 
\[
\Phi_{i,j} \left( M_{1 \rightarrow i} ^{(j)},\ldots, M_{n \rightarrow i} ^{(j)}, Z_i \right) = W_{f_{i,j}}.
\]


For a given cache content placement, a \emph{delivery strategy} is defined as follows.
\begin{definition}[Delivery Strategy]
 By assuming full knowledge of cache contents of all the servers, for each file request $f_{i,j}$,  the Delivery Strategy determines the message set $\left\{ M^{(j)}_{s\rightarrow i} \right\}_{s=1}^n$, for all $j \in [1:R_i]$ and $i \in [1:n]$.
\end{definition}

Now, for each strategy\footnote{We use the terms ``strategy'' and ``delivery strategy'' alternatively.},  we define the following metrics.
\begin{definition}[Communication Cost and Maximum Load]\label{def:Metrics}
~
\begin{itemize}
\item The (normalized) \emph{communication cost} (per request) of a strategy is defined as follows 
\begin{equation}\label{eq:ComCostDef}
C \triangleq \frac{1}{nF}  \sum_{s\in [1:n]} \sum_{i\in [1:n]} \sum_{j\in[1:R_i]}  \dist_G(s,i) H \left( M_{s \rightarrow i}^{(j)} \right).
\end{equation}
\item The (normalized) \emph{maximum load} of a strategy is defined as follows
\[
L \triangleq \frac{1}{F} \max_{s\in [1:n]} \sum_{i\in [1:n]} \sum_{j\in[1:R_i]}  H \left( M_{s \rightarrow i}^{(j)} \right).
\]
\end{itemize}
\end{definition}
It should be noted that since file requests are random, $C$ and $L$ are random variables in the above definitions.

For convenience of readers, Table~\ref{tab:Notations} summarizes important notations used throughout the paper.

\begin{table}
\caption{Notations summary.} \label{tab:Notations}
\centering
\begin{tabular}{|c|c|}
	\hline
	Notation & Description\\
	\hline
	$n$ & Number of servers\\
	\hline
	$K$ & Number of files in the library\\
	\hline
	$\mc{W}=\{W_1,\ldots,W_K\}$ & Files in the library\\
	\hline
	$\mc{P}=\{p_1,\ldots,p_K\}$ & Popularity profile\\
	\hline
	$F$ & File size in bits\\
	\hline
	$M$ & Server's cache size\\
	\hline
	$C$ & Communication cost\\
	\hline
	$L$ & Maximum load\\
	\hline
\end{tabular}
\end{table}

\section{Coded Load Balancing in Grid Networks} \label{sec:CodedLoadBalancingGrid}


In this section, we focus on grid networks where the graph $G$ is a $\sqrt{n} \times \sqrt{n}$ square wireline grid interconnect where $\sqrt{n}$ is assumed to be an integer.

\begin{remark}
	In this section, for the sake of presentation clarity, we may consider a torus with $n$ servers, which helps to avoid boundary effects of the grid, and all the asymptotic results hold for the grid as well.  
\end{remark}

Let us first revisit the baseline schemes with which we compare our proposed scheme's performance. The simplest uncoded scheme which assigns requests to servers is the so-called \emph{nearest replica} strategy, in which files are cached at the servers proportional to their popularity. Then, each request is assigned to the nearest server having a replica of the requested file. More precisely, let us focus on server $i$ at which $R_i$ requests have arrived. For request $j \in \{1,\ldots,R_i\}$, i.e., $W_{f_{i,j}}$, we denote the nearest server to $i$ which has cached this request as $s_j$. Then the delivery strategy will be
\begin{align}
&M_{s_j \rightarrow i}^{(j)} = W_{f_{i,j}}, \\ \nonumber
&M_{s \rightarrow i}^{(j)} = \emptyset \quad \mathrm{for \,\,  all} \quad s \neq s_j.
\end{align}

\begin{remark}
	It should be noted that the nearest replica strategy achieves the minimum communication cost among all \emph{uncoded} schemes. This is true since this scheme minimizes all terms in the summations of \eqref{eq:ComCostDef} (\ie, the definition of communication cost) in Definition~\ref{def:Metrics}.
\end{remark}

While this scheme performs well in terms of communication cost, its maximum load will be of order $\log(n)$ (for more details refer to \cite{SiaPourShar_TPDS17}). The second baseline scheme considered in this paper is the one proposed in \cite{SiaPourShar_TPDS17}, which is based on the concept of power of two choices \cite{ABKU99}. In this scheme, instead of allocating the request to the nearest server having cached that request, the current load of two servers in a limited distance of the request origin is queried, and then the request is assigned to the server with the lower load. In a certain regime of problem parameters, this will result in the maximum load of $\log \log(n)$, which is achieved at the cost of queries' complexity, and higher communication cost compared to the nearest replica strategy \cite{SiaPourShar_TPDS17}.

In contrast, in this section we propose a coded scheme which achieves the maximum load of order $O(1)$, while maintaining almost the same communication cost as the nearest replica strategy, and also does not require querying current load of any servers. To this end, we consider a coded cache placement in the servers using random linear coding, \eg, Fountain-like codes \cite{Luby_2002_LT_Codes, Byers_1998_Fountain_Code, Shokrollahi06_Raptor_Code}. 

Let us first take each file $W_k$ of $F$ bits, and partition it into $\ell$ equal-sized chunks, \ie, $W_k^{(r)}$ where $r\in [1:\ell]$, each of $F/ \ell$ bits. 
Thus, each file $W_k$ becomes the concatenation of the chunks $W_k^{(r)}, r \in [1:\ell]$, \ie,
$W_k = W_k^{(1)} \circ W_k^{(2)} \circ \cdots \circ W_k^{(\ell)}$.
Then, we define a random linear operator $\mc{L}$ as
\begin{equation}\label{eq:CodingOperator}
\mc{L}(W_k) \triangleq \sum_{r=1}^\ell \alpha_r W_k^{(r)},
\end{equation}
where $\alpha_r$'s are constants chosen uniformly at random over a finite field $\mbb{F}_{q}$ with $q = \Theta(2^ {F/\ell})$,
and the summation is over $\mbb{F}_q$ (assuming there is an injective mapping from the file chunks $W_k^{(r)}$ to elements of the finite field $\mbb{F}_q$). Moreover, we assume that each use of the operator $\mc{L}$ is independent of other instances.

In the cache content placement phase, each server stores $\ell M$ \emph{coded chunks} of the files according to the popularity distribution $\mc{P}$ as explained in Algorithm~\ref{alg:Coded_Cache_Placement}. 

\begin{algorithm}[H]
\caption{Coded cache content placement for server $i$}
\label{alg:Coded_Cache_Placement}
\begin{algorithmic}[1]
  \Require $\ell$, $M$, $\mathcal{P}$, and $\mc{W}$
  \Repeat
      \State Sample $k$ according to distribution $\mathcal{P}$
      \State Store a coded chunk $\mc{L}(W_k)$ at server $i$'s cache 
  \Until{server $i$'s cache is full}
\end{algorithmic}
\end{algorithm}

In the content delivery phase, we assume $n$ requests arrive uniformly at random at network servers. Let us consider server $i$ which has $\mathrm{Po}(1)$ requests, \ie, $R_i \sim \mathrm{Po}(1)$, for large $n$. For satisfying each request, $\ell$ nearest 
(in terms of shortest path in the grid network)
coded chunks of the file corresponding to that request should be routed to this server according to Algorithm~\ref{alg:Coded_Content_Delivery}.

\begin{algorithm}[H]
\caption{Coded delivery phase for server $i$}
\label{alg:Coded_Content_Delivery}
\begin{algorithmic}[1]
  \Require $\{ f_{i,j} \}_{j=1}^{R_i}$, $\ell$
  \For{$j=1:R_i$}
      \State $\mc{I} :=$ indices of the $\ell$ nearest servers caching coded chunks corresponding to file $W_{f_{i,j}}$
      \For{$s\in \mc{I}$} 
          \State $M_{s \rightarrow i}^{(j)} :=$ a coded chunk of file $W_{f_{i,j}}$ at server $s$
          \State Forward $M_{s \rightarrow i}^{(j)}$ from server $s$ to server $i$ via the shortest path in the grid network
      \EndFor
  \EndFor
\end{algorithmic}
\end{algorithm}

It is clear that if the field size $q$ is large enough, then Algorithm~\ref{alg:Coded_Content_Delivery} will successfully
send the necessary file chunks to server $i$, so that server $i$ can reconstruct the requested files by its users with the probability of order $1-O(1/q)$, \eg, see \cite[Lemma~1]{SiaFragDig_IT12}.

In the next theorem, we characterize the load balancing performance of the proposed scheme.

\begin{thm}[Maximum Load of Grid]\label{thm:maxloadCode}
	For a $\sqrt{n} \times \sqrt{n}$ grid network,
	suppose that $\mathcal{P}=\{p_1,p_2,\ldots,p_K\}$ be the file popularity distribution and let the number of file chunks be $\ell=\Omega(\log n)$. Then, the maximum load of Algorithm~\ref{alg:Coded_Content_Delivery} is $O(1)$,  w.h.p.
\end{thm}

\begin{proof}
	In Algorithm~\ref{alg:Coded_Content_Delivery}, we have to find $\ell$ coded chunks corresponding to each arriving request. Let us focus on server $u$. Define $\mc{E}_u$ to be the event that for all $k\in [1:K]$, one can find $\ell$ coded chunks corresponding to $W_k$ in $B_{r_k}(u)$ where $r_k \triangleq \sqrt{\alpha \ell/ \pt_k}$ for some positive constant $\alpha$. In above $\pt_k \triangleq 1-(1-p_k)^{M \ell}$ is the probability that an arbitrary server has cached at least a coded chunk of $W_k$. Then, we can state Lemma \ref{lem:Event_E}, which appears after the Theorem proof.

Now, for every fixed server $u\in V(G)$, let us define the indicator random variable $Y_{u, j}$, taking one  if the $j$-th request ($j\in [1:n]$) has asked a coded chunk from server $u$ and zero otherwise.
Notice that a server might be asked to respond a request if the following events occur:
	\begin{itemize}
		\item[(E1)] a request for file $W_k$  is born at a neighborhood of radius $r_k$ of the server, and
		\item[(E2)] a coded chunk of the requested file is cached in the server.
	\end{itemize}
	Thus, we can write
	\begin{align*}
		\Pr{Y_{u, j}=1}  \hspace{-30pt} &  \\
		&=\Pr{Y_{u, j}=1| \mathcal{E}}\Pr{\mathcal{E}}\\ 
		&\quad\quad +  \Pr{Y_{u, j}=1| \neg\mathcal{E}}\Pr{\neg\mathcal{E}}\nonumber\\
		&\stackrel{\text{(a)}}{=} \Pr{Y_{u, j}=1| \mathcal{E}}\Pr{\mathcal{E}}\\ 
		&\quad\quad +  \Pr{Y_{u, j}=1| \neg\mathcal{E}} \times o(1/n)\nonumber\\
		& \le \Pr{Y_{u, j}=1| \mathcal{E}} + o(1/n) \\
		&= \sum_{k=1}^K \Pr{Y_{u, j}=1| \mathcal{E}, W_k \  \mathrm{requested}} p_k + o(1/n)\\
		&\stackrel{\text{(b)}}\le \sum_{k=1}^{K} \frac{|B_{r_k}(u)|}{n}\pt_k\cdot p_k + o(1/n)\\
		& \stackrel{\text{(c)}}{\le} \frac{2\alpha\ell(1+o(1))}{n}+o(1/n)\\
		&= \frac{2\alpha\ell(1+o(1))}{n},
	\end{align*}
	where $\mathcal{E}=\cap_{u\in V(G)}\mc{E}_u$, (a) follows from Lemma~\ref{lem:Event_E}, (b) follows from considering the probabilities of events (E1) and (E2) above, and (c) follows from $B_{r_k}(u)=2r_k^2(1+o(1))$ since $G$ is a grid.
	
	Now, let $S_u=\sum_{j=1}^nY_{u,j}$ count the number of requests that are responded by server $u$, during allocating the total $n$ requests. Hence, we have
	\[
	\Ex{S_u}=\sum_{j=1}^n\Ex{Y_{u,j}}\le 2\alpha\ell(1+o(1)).
	\]
	Applying a Chernoff bound for $S_u$ implies that 
	\[
	\Pr{S_u\geq (1+\delta) 2\alpha\ell} \leq \exp(-\delta^2 \alpha \ell)=o(1/n^2),
	\] 
	for appropriate choices of constants $\delta$ and $\alpha$, and since $\ell=\Omega\left(\log n\right)$.
	Taking union bound over all servers shows that 
	each server is requested at most $O\big(\ell\left(1+o(1)\right)\big)$ times
	where each request involves sending $F/\ell$ bits. Hence, it has to handle at most $O(1)$ bits. This concludes the proof.
\end{proof}

\begin{lemma}\label{lem:Event_E}
	For the event $\mc{E}_u$ defined above, we have
	$\Pr{\mc{E}_u} =1-o(n^{-2})$. Then, the event $\mathcal{E}=\cap_{u\in V(G)}\mc{E}_u$ happens with probability $\Pr{\mc{E}} =1-o(n^{-1})$.
\end{lemma}
\begin{proof}
	For a given $k\in [1:K]$, let $X_{u,k}$ denote an indicator random variable taking $1$ if node $u$ has cached a coded chunk of  $W_k$ and zero otherwise. Thus,
	\[
	\Pr{X_{u,k}=1}=\pt_k=1-(1-p_k)^{M\ell},
	\]
	where according to Algorithm~\ref{alg:Coded_Cache_Placement}, $p_k$ is  the probability that a coded chuck of $W_k$ has been cached on a server. It is easy to see that $Z_{u,k}=\sum_{v\in B_{r_k}(u)}X_{v,k}$ counts the number of servers that have cached $W_k$ in a neighborhood of $u$. 
	For every node $u$, we have 
	\[|B_{r_k}(u)|=\Theta\left(r_k^2\right)=\frac{\alpha  \ell}{\pt_k}.\]
	Since $X_{u,k}$'s are i.i.d random variables and by linearity of expectation, we get that 
	\[\Ex{Z_{u,k}}=r_k^2 \pt_k=\alpha\ell.
	\]
	Provided $\alpha>2$ and $\alpha l/8>\log{n}$ is an appropriate constant,   applying a Chernoff bound results that 
	\begin{align*}
	\Pr{Z_{u,k}<\ell} &\le \Pr{Z_{u,k}\le \Ex{Z_{u,k}}/2}\le \mathrm{e}^{-\Ex{Z_{u,k}}/8}\\
	&= \mathrm{e}^{-\alpha \ell/8}=o(1/{n^3}),
	\end{align*}
	where the last equality follows from  $\ell=\Omega(\log n)$. By applying union bound over all $K\le n$ files in the library, we have 
	\[
	\sum_{k\in[1, K]}\Pr{Z_{u,k}\le \ell}=o(1/{n^2}).
	\]
	This implies that 	$\Pr{\mc{E}_u}=1-o(1/n^2)$ and hence by another application of the union bound we get 
	\begin{align*}
	\Pr{\mc{E}}&=1-\Pr{\cup_{u\in V(G)}\neg \mc{E}_u}\\
	&\ge1-\sum_{u\in V(G)}\Pr{\neg\mc{E}_u}\\
	&=1-o(1/n).
	\end{align*} 
\end{proof}

Next, we analyze the communication cost of the proposed scheme in Algorithm~\ref{alg:Coded_Content_Delivery}. To this end, we first prove a general expression for communication cost of requests with arbitrary popularities in Theorem \ref{thm:comcostCode}. Then, in Corollary \ref{cor:CommCost} we specialize the result of Theorem \ref{thm:comcostCode} to the Zipf popularity profile.

\begin{thm}[Communication Cost] \label{thm:comcostCode}	
For a grid network of size $\sqrt{n} \times \sqrt{n}$,
	suppose that $\mathcal{P}=\{p_1,p_2,\ldots,p_K\}$ be the file popularity distribution, and assume $K=O(n)$ and $\ell=\Omega(\log n)$. 
Define $\pt_k \triangleq 1-(1-p_k)^{M\cdot \ell}$ where we assume $\sqrt{\ell/\pt_k}=o(\sqrt{n})$ for every $k \in [1:K]$. Then, w.h.p., the communication cost for every requested file $W_k$, $1\le k\le K$, is $\Theta\left(\sqrt{\ell/\pt_k}\right)$. Moreover, we have 
\begin{equation}
\Ex{C}=\sum_{k=1}^K \Theta\left(\sqrt{\ell/\pt_k}\right)p_k.
\end{equation}
\end{thm}

\begin{proof}
	{\bf Upper Bound:}
	Suppose that $u\in G$ is an arbitrary server in the grid and  a request for file $W_k$ arrives at server $u$.
	For every server $v\in G$, let $X_{v, k}$ denote the indicator random variable taking $1$ if server $v$ has cached at least one coded chunk of $W_k$, and zero otherwise. Then, for every positive number $r$,
	$Y_u(k,r)=\sum_{v\in B_r(u)}X_{v,k}$  denotes the number of servers that have cached a coded chunk of $W_k$ in $B_r(u)$. 
	Notice that  
	\[
	\tilde{p}_k=\Pr{X_{v,k}=1}=1-\left(1-p_k\right)^{M\cdot \ell}.
	\] 
	Thus,
	\[
	\Ex{Y_u(k,r)}=|B_r(u)|\cdot\pt_k.
	\]
	We know that, for a grid network, $|B_r(u)|=2r(r+1)+1=2r^2(1+o(1))$. Hence by choosing $r_k=\sqrt{z/\pt_k}$, where $z=\max[6\log n, 5\ell]$, we get
	\[
	\Ex{Y_u(k,r_k)}=2z(1+o(1))\geq 12\log n(1+o(1)).
	\]
	Since $X_{v,k}$'s are independent and identical indicator random variables,   
	applying a Chernoff bound 
	for $Y_u(k,r_k)$  yields that 
	\[
	\Pr{Y_u(k,r_k)\leq 0.1\Ex{Y_u(k,r_k)}}=o(1/n^3).
	\]
	Thus, with  probability $1-o(1/n^3)$ 
	\[
	Y_u(k,r_k)\ge 0.1\Ex{Y_u(k,r_k)}=0.2z(1+o(1))\ge \ell,
	\] 
	which means 
	$B_r(u)$ contains at least  $\ell$ servers that have cached a coded chunk of $W_k$. 
	
	Let $\mathcal{E}_{u,k}$ denote the event that server $u$ requests for file $W_k$ and $Y_u(k,r_k)<\ell$.  
	Now by the union bound over all $n$ servers and $K=O(n)$ files we have,
	\[
	\Pr{\cup_{u, k}\mathcal{E}_{u,k}} \le \sum_{u, k}\Pr{\mathcal{E}_{u,k}}=
	\sum_{u, k}o(1/n^3)=o(1/n).
	\] 
	Since the number of bits in each coded chunk is $F/\ell$ and with probability $1-o(1/n)$ we can find all the required coded chunks in $B_r(u)$, the communication cost defined in Definition~\ref{def:Metrics} is $O\left(\sqrt{\ell / \pt_k}\right)$, since $z=O(\ell)$.

	{\bf Lower Bound:}
	If we set $r_k=o (\sqrt{z/\pt_k})$, then $\Ex{Y(k,r_k)}=o(z)=o(\ell)$.
	Thus, by Markov inequality for any constant $\alpha>0$
	\[
	\Pr{Y_u(k,r_k)> \alpha \Ex{Y_u(k,r_k)}}\leq 1/\alpha.
	\] 
	So for every $u$, with probability at least $1-1/ \alpha$, 
	\[
	Y_u(k,r_k) < \alpha \Ex{Y_u(k,r_k)}<\ell.
	\]
	
	Let $T_{u,j}$ denote the indicator random variable taking one if the $j$th request, received by server  $u$, fails to find $\ell$ coded chunks of the requested file in the set $B_r(u)$. Thus, we have 
	$\Pr{T_{u,j}=1}> 1-1/ \alpha$.
	Also let $S=\sum_{u=1}^n \sum_{j=1}^{R_u} T_{u,j}$ denote the total number of failures. Then,   
	\[
	\Ex{S}=n \cdot \Pr{T_{u,j}=1}>(1-1/\alpha)n.
	\]
	Since the requests are independent, another
	application of the Chernoff bound for random variable $S$ results that  
	\[
	\Pr{S<\Ex{S}/2}<\exp({-\Omega(n)}).
	\]
	Hence, w.h.p. 
	at least $(1-1/\alpha)n/2$ requests 
	cannot be responded in the  $r$-neighborhood of the requesting server, and the communication cost for $W_k$ is
	$\Omega(\sqrt{\ell/\pt_k})$. 
	
	Since the upper and lower bounds meet, the communication cost of requesting file $W_k$ is $\Theta(\sqrt{\ell/\pt_k})$.	
	Consequently, by averaging over the library with probability distribution $\mc{P}$, the average communication cost is
	\[
	\Ex{C}=\sum_{k=1}^K \Theta\left(\sqrt{\ell/\pt_k} \right) p_k.
	\]
\end{proof}

The following corollary of Theorem~\ref{thm:comcostCode} characterizes the communication cost of our problem under the Zipf popularity profile.

\begin{cor}\label{cor:CommCost}
For a grid network of size $\sqrt{n} \times \sqrt{n}$,
suppose that the cache size $M$ of each server  is a constant, the number of files is $K=n^{\delta}$, for any $\delta\in (0,1]$, and the number of file chunks is $\ell=\Theta(\log n)$. Then, for Zipf popularity distribution with parameter $\gamma$, the average communication cost is
	\begin{equation*}
\Ex{C} = \left\{
\begin{array}{llll}
\Theta\left(\sqrt{K/M}\right) &: &\quad 0\leq \gamma <1, \\
\Theta\left(\sqrt{K/M\log K}\right) &: &\quad  \gamma=1,\\
\Theta\left({K^{1-\gamma/2}/\sqrt{M}}\right) &: &\quad  1<\gamma<2,\\
\Theta\left(\log K/\sqrt{M}\right) &: &\quad  \gamma=2,\\
\Theta(\sqrt{\ell}) &: &\quad  \gamma>2.
\end{array}
\right.	
\end{equation*} 
\end{cor}

\begin{proof}
	To show this, for some small constant $\epsilon>0$, let us define two sets, say $A \triangleq \{k : 1\le k\le K,~  p_k M\ell>\epsilon\}$ and its complement $\bar{A}$. So for every $k\in A$, we have $\pt_k=\Theta(1)$ and for every $k\in \bar{A}$,  
	\[
	\pt_k=1-(1-p_k)^{M\ell}\approx p_k M\ell.
	\] 
	Thus, the average communication cost reduces to the following summation 
	\[
	\Ex{C}=\sum_{k\in A} \Theta\left(\sqrt{\ell} \right)p_k +\sum_{k\in \bar{A}} \Theta\left(\sqrt{p_k/M}\right).
	\]
	For Zipf distribution, $p_k$'s are decreasing in $k$ so let us define $k^*=\max A$ if $A$ is not empty and $k^*=0$, otherwise. Also let $p_0=0$. Then,
	\begin{align}\label{averagecc}
	\Ex{C}=
	\underbrace{\Theta(\sqrt{\ell}){\sum_{k=0}^{k^*}p_k}}_{S_1(\gamma)}+ \underbrace{ \sum_{k=k^*+1}^K \Theta\left(\sqrt{p_k/M}\right)}_{S_2(\gamma)}.
	\end{align}
	Now let us estimate  $S_1(\gamma)$ in (\ref{averagecc}). It is clear that 
	$0\leq \sum_{k=0}^{k^*}p_k<1$, thus
	\begin{equation}\label{S1estim}
	S_1(\gamma) = \left\{
	\begin{array}{llll}
	O(\sqrt{\ell}) &: &\quad 0\le \gamma \le 1, \\
	\Theta (\sqrt{\ell}) &: &\quad \gamma>1,
	\end{array}
	\right.	
	\end{equation}
	where $S_1(\gamma)=\Theta(\sqrt{\ell})$ as $p_1=\Theta(1)$ when $\gamma>1$.
	In what follows we provide an estimation for $S_2(\gamma)=\sum_{k=k^*+1}^K \Theta(\sqrt{p_k/M})$.
	Let us define $\Lambda(\gamma, s)\triangleq \sum_{k=s+1}^K k^{-\gamma}$. So
	\begin{align*}
	S_2(\gamma) &= \sum_{k=k^*+1}^K\Theta\left(\frac{k^{-\gamma/2}}{\sqrt{M\Lambda(\gamma, 0)}}\right) 
	=\Theta\left(\frac{\Lambda(\gamma/2, k^*)}{\sqrt{M\Lambda(\gamma, 0)}}\right).
	\end{align*}
	
	Next, in order to proceed, we need Lemma \ref{lem:LambdaBounds} stated after the Theorem's proof. From Lemma~\ref{lem:LambdaBounds}, we know that
	
	\begin{equation*}
	\Lambda(\gamma, k^*) = \left\{
	\begin{array}{llll}
	\Theta\left(K^{1-\gamma}\right)-\Theta({k^*}^{1-\gamma}) &: &\quad 0 \leq \gamma <1, \\
	\Theta\left({\log K}\right)-\log k^* &: &\quad  \gamma=1, \\
	O\left(1\right) &: &\quad \gamma>1.
	\end{array}
	\right.	
	\end{equation*}
	Note that by the definition of $k^*$, we have
	\[
	\frac{\epsilon}{M\ell}\le p_{k^*}< {k^*}^{-\gamma},
	\]
	and hence	$k^*\le \left(\frac{M\ell}{\epsilon}\right)^{1/\gamma}=O((\log n)^{1/\gamma})$. This implies that we can ignore terms ${k^*}^{1-\gamma}$ and $\log k^*$ in comparison with $\Theta(K^{1-\gamma})$ and $\Theta(\log K)$, respectively,  since $K=n^{\delta}$, for some constant $\delta\in(0,1]$.
	In other words we can write
	\begin{align*}
	S_2(\gamma) &=\Theta\left(\frac{\Lambda(\gamma/2, 0)}{\sqrt{M\Lambda(\gamma, 0)}}\right).
	\end{align*}
	Then, we will have
	\begin{equation}\label{s2estim}
	S_2(\gamma) = \left\{
	\begin{array}{llll}
	\Theta\left(\sqrt{K/M}\right) &: &\quad 0\leq \gamma <1, \\
	\Theta\left(\sqrt{K/M\log K}\right) &: &\quad  \gamma=1,\\
	\Theta\left({K^{1-\gamma/2}/\sqrt{M}}\right) &: &\quad  1<\gamma<2,\\
	\Theta\left(\log K/\sqrt{M}\right) &: &\quad  \gamma=2,\\
	O(1) &: &\quad  \gamma>2.
	\end{array}
	\right.	
	\end{equation} 
	Now, considering equalities (\ref{S1estim}) and (\ref{s2estim}) completes the proof.
\end{proof}

	\begin{lemma}\label{lem:LambdaBounds}
		For $\Lambda(\gamma, k)=\sum_{l=k}^K l^{-\gamma}$, where $\gamma\ge 0$, we have
		\[
		\Lambda(\gamma, k) = \left\{
		\begin{array}{llll}
		\Theta\left(K^{1-\gamma}\right)-\Theta({k}^{1-\gamma}) &: &\quad 0 \leq \gamma <1, \\
		\Theta\left({\log K}\right)-\log k &: &\quad  \gamma=1, \\
		O\left(1\right) &: &\quad \gamma>1.
		\end{array}
		\right.	
		\]
	\end{lemma}
	\begin{proof}
		First notice that by the definition of Riemann integration, we can upper and lower bound $\Lambda(\gamma,k)$ as follows
		\begin{align*}
		\int_{k}^{K+1} t^{-\gamma}dt \le \sum_{l=k}^K l^{-\gamma} \le k^{-\gamma} + \int_{k}^{K} t^{-\gamma}dt.
		\end{align*}
		This can be simplified to
		\begin{align*}
		\frac{1}{1-\gamma}\left[ (K+1)^{1-\gamma} - k^{1-\gamma} \right] \le\Lambda(\gamma,k)
		\end{align*}
		and
		\begin{align*}
		\Lambda(\gamma,k) \le k^{-\gamma}+ \frac{1}{1-\gamma}\left[ K^{1-\gamma} - k^{1-\gamma} \right].
		\end{align*}
		
		Now, considering the following three cases $0\le \gamma < 1$, $\gamma = 1$, and $\gamma > 1$ will conclude the proof.
	\end{proof}

It is interesting to note that the nearest replica strategy also arrives at almost the same communication cost derived in Corollary~\ref{cor:CommCost} (see \cite[Theorem~3]{SiaPourShar_TPDS17}).

\section{Coded Load Balancing in General Networks} \label{sec:CodedLoadBalancingGeneral}
The load balancing behavior of our proposed algorithm in last section, which was proved for grid networks, can be generalized for other networks as well. In this section, we will generalize our results for other networks with some symmetry properties. In order to do this, in this paper, we define \textit{Locally Symmetric Networks} as follows.

\begin{definition}[Locally Symmetric] \label{def:LocallySymmetric}
	We define a class of graphs $\mathcal{G}$ \emph{locally symmetric}, if for every $G\in \mathcal{G}$, for an arbitrary integer $r\ge 1$, and for every $u, v\in V(G)$, we have
		\[
		\max\left[ \frac{|B_{r}(u)|}{|B_{r}(v)|}, \frac{|B_{r}(v)|}{|B_{r}(u)|}\right]=O(1).
		\]
	Then, a graph is called locally symmetric if it belongs to such a class.	
\end{definition}
As will be proved later, \emph{random geometric graphs} and \emph{random regular graphs} are two examples of locally symmetric networks, w.h.p.

To generalize our results, here we need a more general notion of $r_k$ which was introduced in the previous section as the search distance for finding $\ell$ coded chunks of $W_k$. Here, for a given network $G$ and  $k\in [1:K]$, we define $r_k$ as the smallest integer such that for every node $u$ in $G$, we have
\[
\big|B_{r_k}(u)\big|\ge \frac{32 \log n}{\pt_k}=\frac{\alpha_k\ell}{\pt_k},
\]
for some constants $\alpha_k$, $k\in[1:K]$.

We also need to slightly modify Algorithm~\ref{alg:Coded_Content_Delivery} for locally symmetric class of graphs, as stated in Algorithm~\ref{alg:Coded_Content_Delivery for GP}.

\begin{algorithm}[H]
	\caption{Coded delivery phase for server $i$ in locally symmetric graphs }
	\label{alg:Coded_Content_Delivery for GP}
	\begin{algorithmic}[1]
		\Require $\{ f_{i,j} \}_{j=1}^{R_i}$, $\ell$, $\{r_1,r_2\ldots,r_k\}$
		\For{$j=1:R_i$}
		\State $\mc{J} :=$ set of servers that have cached a coded chunk of $W_{f_{i ,j}}$ at distance at most $r_k$ from server $i$
		\State $\mc{I} :=$ set of $\ell$ servers that are selected uniformly at random from $\mc{J}$
		\For{$s\in \mc{I}$} 
		\State $M_{s \rightarrow i}^{(j)} :=$ a coded chunk of file $W_{f_{i,j}}$ at server $s$
		\State Forward $M_{s \rightarrow i}^{(j)}$ from server $s$ to server $i$ via the shortest path in the network
		\EndFor
		\EndFor
	\end{algorithmic}
\end{algorithm}

\begin{remark}
	Notice that for grid networks, Algorithm~\ref{alg:Coded_Content_Delivery for GP} is reduced to Algorithm~\ref{alg:Coded_Content_Delivery}. This happens because the sets $\mc{I}$ and $\mc{J}$ become the same w.h.p.
\end{remark}

The load balancing behavior of Algorithm~\ref{alg:Coded_Content_Delivery for GP} for locally symmetric networks is characterized in the following theorem.

\begin{thm}[Maximum Load of Locally Symmetric Graphs]\label{thm:maxloadCode_Gen}
	Suppose that 
	$\mathcal{P}=\{p_1,p_2,\cdots,p_K\}$ be the file popularity distribution and let the number of file chunks be $\ell=\Omega(\log n)$. Also, assume that 
	$G=(V, E)$ be a locally symmetric network. Then, the maximum load achieved by Algorithm \ref{alg:Coded_Content_Delivery for GP} is $O(1)$, w.h.p.
\end{thm}

\begin{proof}
	
	For every node $u\in V(G)$, let $\mc{E}_u$ denote the event that for every $k\in[1:K]$, there exist at least $|B_{r_k}(u)|\pt_k/2$ coded chunks of $W_k$ in $B_{r_k}(u)$. Then, we have Lemma \ref{lem:NewEvent_E} stated after the Theorem's proof, whose proof is similar to that of Lemma \ref{lem:Event_E}.

	Similar to the proof of Theorem~\ref{thm:maxloadCode}, by defining $Y_{u, j}$ (taking one  if the $j$-th request, $j\in [1:n]$, has asked a coded chunk from server $u$ and zero otherwise), we can write
	\begin{align*}
		\Pr{Y_{u, j}=1}  \hspace{-30pt} &  \\
		&=\Pr{Y_{u, j}=1| \mathcal{E}}\Pr{\mathcal{E}}\\ 
		&\quad\quad +  \Pr{Y_{u, j}=1| \neg\mathcal{E}}\Pr{\neg\mathcal{E}}\nonumber\\
		& \stackrel{\text{(a)}}{\le} \Pr{Y_{u, j}=1| \mathcal{E}} + o(1/n) \\
		&= \sum_{k=1}^K \Pr{Y_{u, j}=1| \mathcal{E}, W_k \  \mathrm{requested}} p_k + o(1/n)\\
		&\stackrel{(b)}\le \sum_{k=1}^{K} \sum_{w\in B_{r_k}(u)}\frac{\pt_k\cdot p_k}{n}\cdot \frac{\ell}{X_{w,k}} + o(1/n)\\
		& \stackrel{\text{(c)}}{\le}  \sum_{k=1}^{n}\sum_{w\in B_{r_k}(u)}\frac{2p_k\ell}{n|B_{r_k}(w)|} +o(1/n),
	\end{align*}
	where $\mathcal{E}=\cap_{u\in V(G)}\mc{E}_u$, (a) follows from Lemma \ref{lem:NewEvent_E}, in (b) $X_{w,k}$ is defined to be the number of servers in $B_{r_k}(w)$ that have cached a coded chunk of $W_k$,  and (c) follows from the fact that conditioned on $\mathcal{E}$, we have $X_{k,u}\ge |B_{r_k}(u)|\cdot\pt_k/2$. Thus, we have 
	\begin{align*} 
		\Pr{Y_{u, j}=1} &\le  \sum_{k=1}^{n}\sum_{w\in B_{r_k}(u)}\frac{2p_k\ell}{nm} +o(1/n)\\ 
		&=  \sum_{k=1}^{n}|B_{r_k}(u)|\frac{2p_k\ell}{nm} +o(1/n) \\
		&\stackrel{(a)}\le \sum_{k=1}^{n}\frac{2c\ell p_k}{n}+o(1/n)\\
		&\le\frac{2c\ell+1}{n} 
	\end{align*}
	where 
\[
m\triangleq\min_{w\in B_{r_k}(u)} |B_{r_k}(w)|,
\]	
	and in (a) we have used the fact that if $G$ is locally symmetric, then for some constant $c$ we have $|B_{r_k}(u)|\le c m$.

	Now, let $S_u=\sum_{j=1}^nY_{u,j}$ count the number of requests that are responded by server $u$, during allocating $n$ requests. Hence, we have
	\[
	\Ex{S_u}=\sum_{j=1}^n\Ex{Y_{u,j}}\le 2 c\ell+1.
	\]
	Applying a Chernoff bound for $S_u$ implies that 
	\[
	\Pr{S_u\geq (1+\delta) 2 c\ell} \leq \exp(-\delta^2 c \ell)=o(1/n^2),
	\] 
	for appropriate choice of constant $\delta$, and since $\ell=\Omega\left(\log n\right)$. Thus, taking union bound over all servers shows that 
	each server is requested at most $O\big(\ell\big)$ times
	where each request involves sending $F/\ell$ bits. Hence, it has to handle at most $O(1)$ bits. This concludes the proof.
\end{proof}

\begin{lemma}\label{lem:NewEvent_E}
	For the event $\mc{E}_u$ defined above, we have
	$\Pr{\mc{E}_u} =1-o(n^{-2})$. Then, the event $\mathcal{E}=\cap_{u\in V(G)}\mc{E}_u$ happens with probability $\Pr{\mc{E}} =1-o(n^{-1})$.
\end{lemma}

To apply the above result for other network topologies, it is sufficient to show that those networks are locally symmetric. Here, we consider three well-studied families of networks, namely, hypercube, random geometric graphs and random regular graphs.
 
\textbf{Hypercube:}
 A hypercube, also called $n$-cube, is a network where the nodes can be presented by $n$-bit binary words.  Two nodes are connected if and only if  the nodes only differ in one bit. 
 
 \begin{proposition}\label{pro:hyper}
 	For every given $n>0$, $n$-cube is locally symmetric. 
 \end{proposition} 
 \begin{proof}
 	By the definition of the $n$-cube, two nodes are connected if and only if they differ in only one bit.
 	It is not hard to see that if $u$ is adjacent to $v$, then for every $n$-bit $x$, $u\oplus x$ is also adjacent to $v\oplus x$. Note that we use ``$\oplus$'' to denote the XOR of two words. Thus,
 	for each node $u$ and given an integer $r\ge 1$, we have 
 	$|B_{r}(u)|=\sum_{i=0}^r {n\choose i}$,
 	which satisfies the locally symmetric graph constraint defined in Definition~\ref{def:LocallySymmetric}.
\end{proof}

\textbf{Random Geometric Graphs:}
A random geometric graph, denoted by $\text{RG}(n, \lambda)$, is a network with $n$ nodes chosen uniformly at random in the square $[0,\sqrt{n}]\times [0,\sqrt{n}]$. Two nodes are connected if and only if their Euclidean distance is at most $\lambda$. It is well-know (e.g., see \cite{diaz2016}) that if we set $\lambda>\lambda_c=\sqrt{\frac{\log n+O(1)}{\pi}}$, then, the obtained graph is connected, w.h.p.

\begin{proposition}\label{pro:RandomGeometric}
	For a given $n$, a typical $n$-node random geometric network is locally symmetric, w.h.p.
\end{proposition}
\begin{proof}
	To show the proposition, we apply a result  from \cite{diaz2016} showing that if $G$ is a realization of $\text{RG}(n,\lambda)$ with $\lambda=\omega(\sqrt{\log n})$, then for every $u, v\in V(G)$ we have
	
	\[
	\frac{\dist_{\mathsf{E}}(u, v)}{\lambda}\le \dist_G(u,v)\le \frac{\dist_{\mathsf{E}}(u,v)}{\lambda}\big(1+o(1)\big)
	\]
	w.h.p., where $\dist_{\mathsf{E}}(\cdot,\cdot)$ denotes the Euclidean distance between $u$ and $v$.  By applying the above inequalities,  we have $v\in B_r(u)$ if and only if 
	\[
	r\lambda\le \dist_{\mathsf{E}}(u, v)\le r\lambda \big(1+o(1) \big).
	\]  
	This implies that for each $u\in V(G)$,
	\[
	D_{r\lambda}(u)\subseteq B_r(u)\subseteq D_{r\lambda(1+o(1))}(u),
	\]
	where $D_{l}(u)$ is a subset of $V(G)$ whose Euclidean distances 
	from $u$ is at most $l$. 
	Clearly, one may see that the probability that one node falls in $D_l(u)$ is $\frac{\pi l^2}{n}$, where $l^2\pi$ is the area of the disk $D_l(u)$.
	Since nodes are drawn uniformly and independently at random from $[0, \sqrt{n}]\times[0, \sqrt{n}]$, the expected number of nodes falls in 
	$D_l(u)$ is $l^2\pi$. Applying a Chernoff bound, we conclude that the number of nodes in $D_l(u)$ is $\Theta(l^2\pi)$. By setting $l$ to be either $r{\lambda}$ or $r\lambda(1+o(1))$, we get that 
	$|B_r(u)|=\Theta(r^2\lambda^2)$. Therefore, for a typical random geometric graph $G$, $G$ is locally symmetric. 
\end{proof}

\textbf{Random $d$-Regular Graphs:}
Suppose that for every $n$ and $d\ge 3$, $\mc{G}_{n,d}$ is the family of all $d$-regular graphs with $n$ nodes. Assume that  $G$ is a randomly chosen graph from $\mc{G}_{n,d}$. Then, it is well-known  that w.h.p., for every $r=o(\log_d n)$,
and  $u\in V(G)$,   the subgraph induced by set $B_r(u)$ looks like a $d$-ary tree rooted at $r$. Hence, we have 
$|B_{r}(u)|=\Theta\Big( d(d-1)^{r-1} \Big)$,
which holds for every $u\in V(G)$ and $r=o(\log_d n)$ (for more details refer to \cite{lubetzky2010}). Thus, it is clear that $G$ is locally symmetric for $r=o(\log_d n)$ which is sufficient for Theorem~\ref{thm:maxloadCode_Gen} to be valid for random $d$-regular graphs.


\section{Numerical Analysis}\label{sec:Simulations}

In this section, we use Monte Carlo simulations to investigate the communication cost and maximum load performance of a content delivery network, under the model described in Section \ref{sec:SystemModel}. Our simulator has been written in Python and can run multiple instances of simulations in parallel \cite{CLB-Simulator-GitHub_2018}.
Here, our main goal is twofold. First, we verify our key understandings from asymptotic theoretical results about the benefits of coding. Second, we investigate other aspects of the proposed scheme not revealed in the asymptotic analyses.

In order to clarify the effect of different parameters on the network performance, we start numerical analysis with a uniform popularity distribution, \ie, $\gamma=0$ for the Zipf distribution, and then investigate the effect of varying $\gamma$ on the results.
In all of the following simulations, each data point is obtained by taking average over 5000 simulation runs.

In Fig.~\ref{fig:SrvSzVar_MaxLoad_Lattice_Uniform_Coded_fn=100_chnmax=1_itr=5000_Fig1} the maximum load $L$ is plotted versus the network size $n$ for the \emph{nearest replica}, \emph{power of two choices}, and \emph{coded load balancing} strategies. Here, the network topology is a grid, the library size is $K=100$ and we assume a uniform popularity profile. 
Note that the maximum load of power of two choices scheme is better than the coded scheme when chunk size is $\ell=4$. However, if chunk size is increased to $\ell=10$, our proposed scheme surpasses all baseline schemes.
Also, Fig.~\ref{fig:SrvSzVar_MaxLoad_Lattice_Uniform_Coded_fn=100_chnmax=1_itr=5000_Fig2} shows similar trend for the cache size $M=10$.

\begin{figure}
\begin{center}
\includegraphics[width=.47\textwidth]{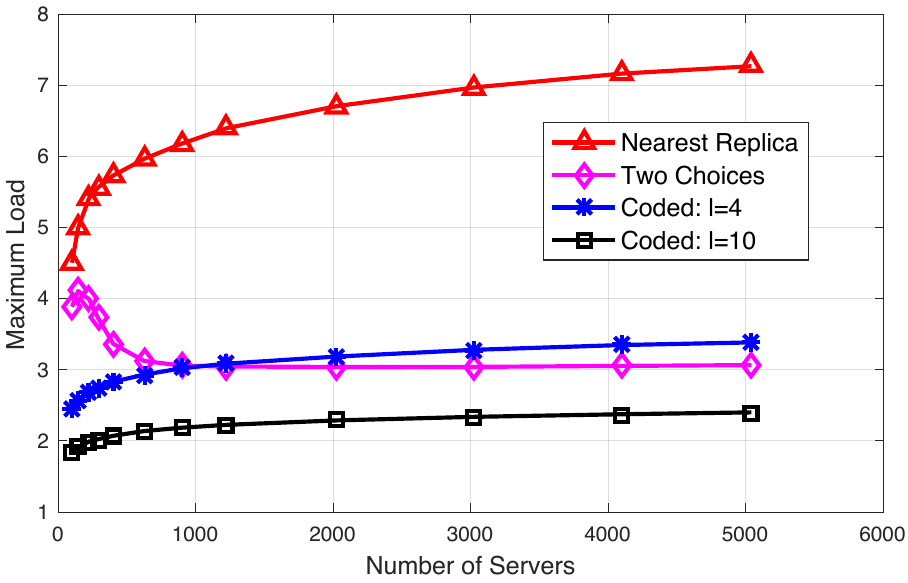}
\end{center}
\caption{Comparing the maximum load $L$ of the nearest replica, power of two choices, and coded load balancing strategies versus the number of servers $n$. The library size is $K=100$ and the cache size is $M=2$.}
\label{fig:SrvSzVar_MaxLoad_Lattice_Uniform_Coded_fn=100_chnmax=1_itr=5000_Fig1}
\end{figure}

\begin{figure}
	\begin{center}
		\includegraphics[width=.47\textwidth]{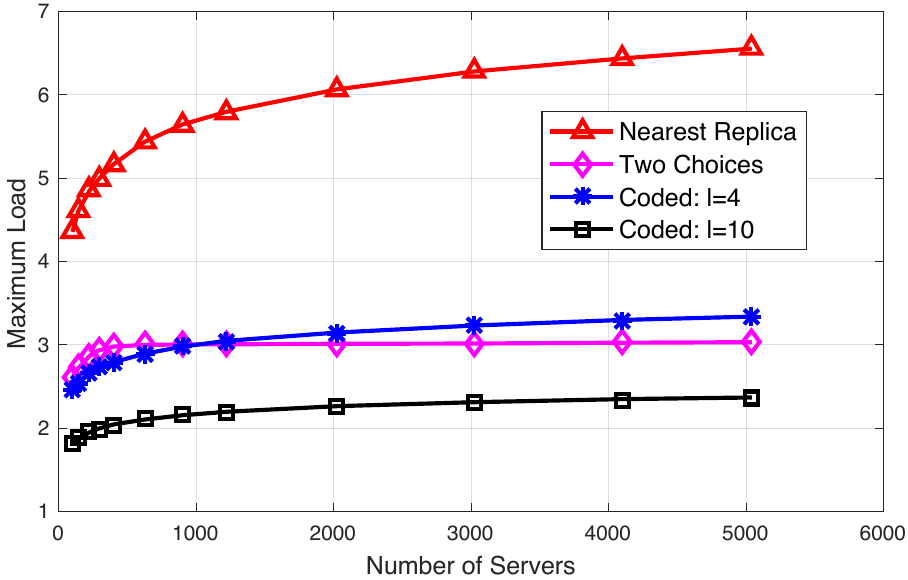}
	\end{center}
	\caption{Comparing the maximum load $L$ of the nearest replica, power of two choices, and coded load balancing strategies versus the number of servers $n$. The library size is $K=100$ and the cache size is $M=10$.}
	\label{fig:SrvSzVar_MaxLoad_Lattice_Uniform_Coded_fn=100_chnmax=1_itr=5000_Fig2}
\end{figure}

In order to investigate the role of chunk size $\ell$ in the performance of proposed scheme, in Fig.~\ref{fig:ChnkSzVar_MaxLoad_Lattice_Uniform_Coded_sn=3025_fn=100_cs=10_chnmax=1_itr=5000}, we have considered a grid network of size $n=1024$, a library size of $K=100$, and  cache sizes $M\in\{1,10\}$.
For comparison, the results of nearest replica and power of two choices strategies are also highlighted in this figure. As it is observed in Fig.~\ref{fig:ChnkSzVar_MaxLoad_Lattice_Uniform_Coded_sn=3025_fn=100_cs=10_chnmax=1_itr=5000}, in order to obtain the coding benefit (\ie, surpassing the power of two choices performance), the chunk size should be above a threshold.

\begin{figure}
	\begin{center}
		\includegraphics[width=.47\textwidth]{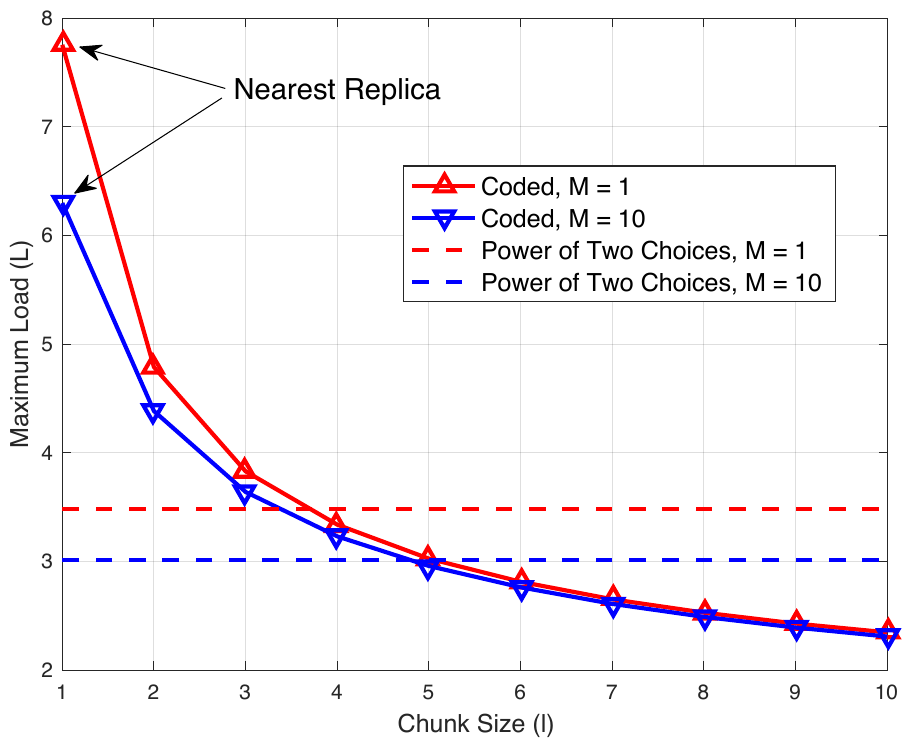}
	\end{center}
	\caption{The maximum load $L$ of the proposed coded load balancing strategy versus the chunk size $\ell$. The library size is $K=100$ and the number of servers is $n=1024$.}
	\label{fig:ChnkSzVar_MaxLoad_Lattice_Uniform_Coded_sn=3025_fn=100_cs=10_chnmax=1_itr=5000}
\end{figure}

In Fig.~\ref{fig:CacheSzVar_MaxLoad_Lattice_Uniform_Coded_sn=1024_fn=100_chnmax=1} the maximum load is depicted versus the cache size $M$ of each server. Here, the network topology is grid, $n=1024$, and $K=100$. It is interesting to note that for large enough cache size, the maximum load does not decrease as cache size increases.
\begin{figure}
	\begin{center}
		\includegraphics[width=.47\textwidth]{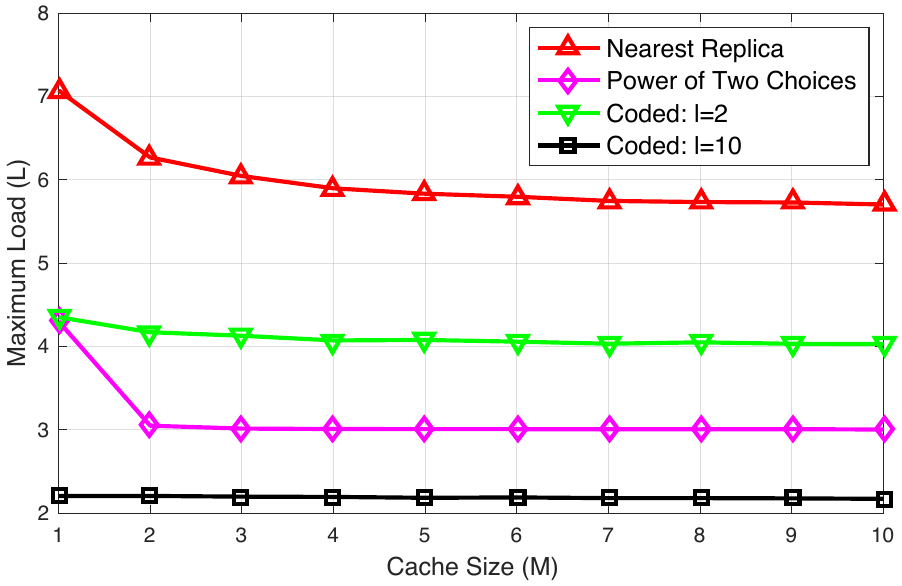}
	\end{center}
	\caption{The maximum load $L$ of proposed coded load balancing strategy versus the cache size $M$. The library size is $K=100$ and the number of servers is $n=1024$.}
	\label{fig:CacheSzVar_MaxLoad_Lattice_Uniform_Coded_sn=1024_fn=100_chnmax=1}
\end{figure}

Thus far, for the purpose of clarity, we have investigated different aspects of the proposed coded strategy in the case of $\gamma=0$. Here, we move forward to investigate the role of Zipf parameter $\gamma$ on the network performance as presented in Fig.~\ref{fig:ZipfGammaVar_MaxLoad_Lattice_Zipf_Coded_sn=1024_fn=100_cs=2_chnmax=1}, in which the maximum load is plotted versus $\gamma$. We observe that the maximum load does not depend on the popularity profile (here characterized by $\gamma$), which is consistent with our finding in Theorem~\ref{thm:maxloadCode}. This is a direct consequence of our \emph{proportional} cache content placement described in  Algorithm~\ref{alg:Coded_Cache_Placement}.
\begin{figure}
		\centering
		\includegraphics[width=.47\textwidth]{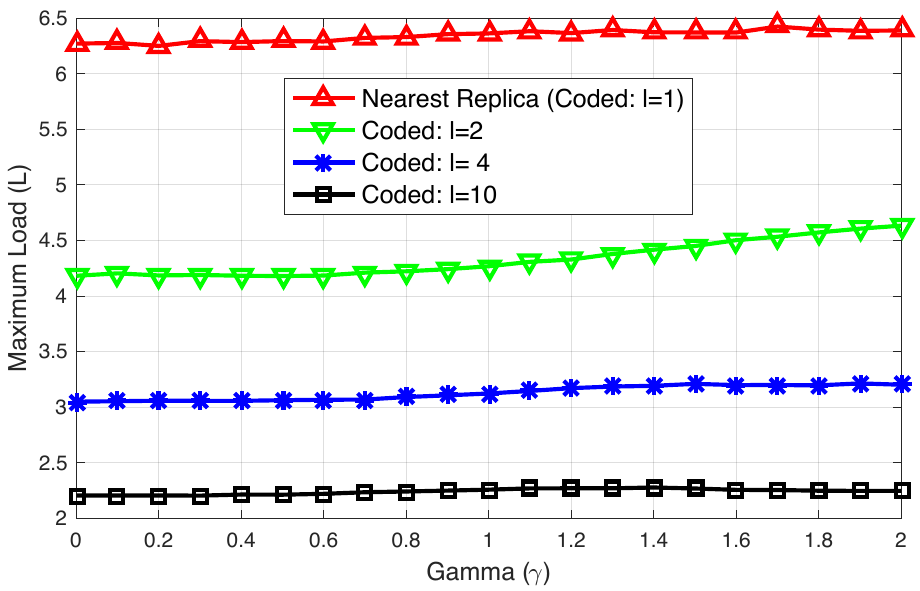}
		\caption{Maximum load of the coded load balancing scheme versus the Zipf parameter $\gamma$ for $n=1024$ and $M=2$. The network topology is a grid and the file popularity is uniform.}
		\label{fig:ZipfGammaVar_MaxLoad_Lattice_Zipf_Coded_sn=1024_fn=100_cs=2_chnmax=1}
\end{figure}

In order to further investigate the validity of our results for other network topologies, in Fig.~\ref{fig:SrvSzVar_Maxload_RGG_Uniform_Coded_fn=100_cs=2_chn=1_chnmax=1} we have plotted the maximum load versus number of servers for random geometric graph (RGG)\footnote{Here, we consider a disk model RGG.}. Comparing this figure with Figure~\ref{fig:SrvSzVar_MaxLoad_Lattice_Uniform_Coded_fn=100_chnmax=1_itr=5000_Fig1} shows that our proposed method performs similarly for grid and random geometric networks.
\begin{figure}
	\centering
	\includegraphics[width=.47\textwidth]{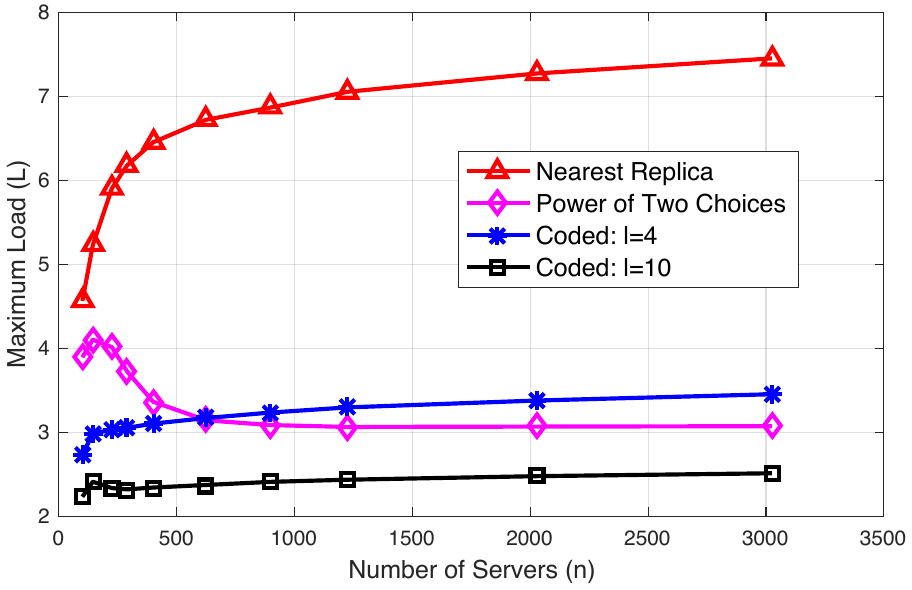}
	\caption{Maximum load of the coded load balancing scheme versus the number of servers $n$ for the random geometric graph (RGG). The library size is $K=100$, cache size is $M=2$, and the popularity profile is uniform.}
	\label{fig:SrvSzVar_Maxload_RGG_Uniform_Coded_fn=100_cs=2_chn=1_chnmax=1}
\end{figure}


In all above simulation scenarios, we have verified the superiority of our result compared to previous schemes in terms of maximum load. 
Here, we compare the communication cost of the proposed scheme with the nearest replica (uncoded) scheme which has the minimum communication cost among previous schemes\footnote{The power of two choices scheme has always higher communication cost than the nearest replica scheme and hence is not plotted in Fig.~\ref{fig:CacheSzVar_ComCost_Lattice_Uniform_Coded_sn=1024_fn=100_chnmax=1} and Fig.~\ref{fig:ZipfGammaVar_ComCost_Lattice_Zipf_Coded_sn=1024_fn=100_cs=2_chnmax=1}.}.
To this end, Fig.~\ref{fig:CacheSzVar_ComCost_Lattice_Uniform_Coded_sn=1024_fn=100_chnmax=1} shows the communication cost versus cache size $M$ for a grid topology of $n=1024$ nodes, library size $K=100$, and uniform popularity profile. As it shows, the proposed coded scheme results in a slightly lower communication cost compared to the nearest replica strategy.
\begin{figure}
	\begin{center}
		\includegraphics[width=.47\textwidth]{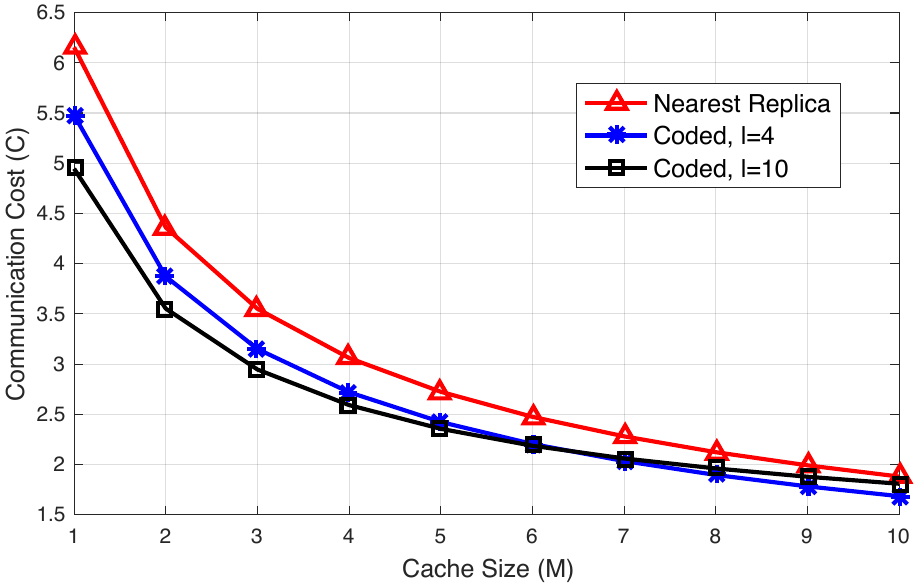}
	\end{center}
	\caption{Communication cost versus cache size $M$. The network topology is a grid of size $n=1024$, $K=100$, and popularity profile is uniform.}
	\label{fig:CacheSzVar_ComCost_Lattice_Uniform_Coded_sn=1024_fn=100_chnmax=1}
\end{figure}

Finally, in Fig.~\ref{fig:ZipfGammaVar_ComCost_Lattice_Zipf_Coded_sn=1024_fn=100_cs=2_chnmax=1} communication cost is plotted versus $\gamma$. As it is observed from this figure, communication cost of all the schemes show a decreasing trend as $\gamma$ increases. It is interesting to note that the coded schemes' communication cost is better in the lower $\gamma$ regime compared to the nearest replica scheme.
\begin{figure}
		\centering
		\includegraphics[width=.47\textwidth]{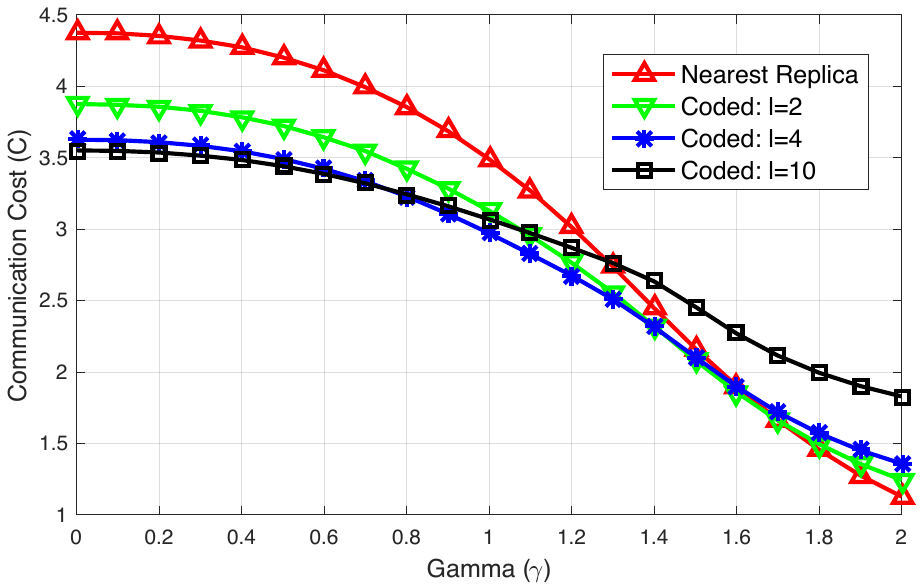}
	    \caption{Communication cost of the coded load balancing scheme versus the Zipf parameter $\gamma$ for the same set of parameters stated in Fig.~\ref{fig:ZipfGammaVar_MaxLoad_Lattice_Zipf_Coded_sn=1024_fn=100_cs=2_chnmax=1}.}
		\label{fig:ZipfGammaVar_ComCost_Lattice_Zipf_Coded_sn=1024_fn=100_cs=2_chnmax=1}
\end{figure}

\section{Discussions and Concluding Remarks}\label{sec:Conclusions}
We have proposed and investigated a coded cache content placement and content delivery scheme which is shown to surpass the \emph{nearest replica strategy} and \emph{power of two choices} baseline schemes in terms of load balancing performance. By deriving closed-form expressions for a grid network, we have shown that the proposed scheme will result in an almost perfect load balancing performance without sacrificing communication cost (Table~\ref{tab:ResultSummary} summarizes our results versus results of the two baseline schemes proposed in \cite{SiaPourShar_TPDS17}).
Furthermore, we have generalized the above result for a more general class of networks, including Hypercube, Random Geometric Graphs, and Random $d$-Regular Graphs.
By performing extensive simulations, we have verified our theoretical findings as well as investigated the non-asymptotic performance of the proposed scheme.

\begin{table*}
\begin{center}
\caption{The scaling results for the proposed coding scheme versus the baseline schemes of \cite{SiaPourShar_TPDS17} (for grid topology with uniform file popularity).} \label{tab:ResultSummary}	

\begin{tabular}{|c|c|c|c|c|}
	\hline
	& $L$ & Regime ($L$) & $C$ & Regime ($C$)\\ 
	\hline
	Nearest replica \cite{SiaPourShar_TPDS17} & $\Theta(\log{n})$ & $K=n^\delta$ for $0<\delta<1$, $M=\Theta(1)$ & $\Theta\left(\sqrt{\frac{K}{M}} \right)$ & $M \ll K$ \\
	\hline
	Power of two choices \cite{SiaPourShar_TPDS17} & $\Theta(\log\log{n})$ & \parbox[c]{5cm}{$K=n$, $M=n^\alpha$ and $r=n^\beta$ where $\alpha+2\beta\ge 1+2\frac{\log\log{n}}{\log{n}}$ for $0<\alpha,\beta<1/2$} & $\Theta(r)$ & \parbox[c]{5cm}{$K=n$, $M=n^\alpha$ and $r=n^\beta$ where $\alpha+2\beta\ge 1+2\frac{\log\log{n}}{\log{n}}$ for $0<\alpha,\beta<1/2$}\\
	\hline
	Coded & $O(1)$ & $\ell=\Omega(\log{n})$ & $\Theta\left(\sqrt{\frac{K}{M}} \right)$ & \parbox[c]{5cm}{$M=\Theta(1)$, $K=n^\delta$, for $0<\delta<1$, and $\ell=\Theta(\log{n})$}\\
	\hline
\end{tabular}
\end{center}
\end{table*}

Finally, here we comment on the complexity of coding/decoding of the proposed scheme. In the proposed coded caching scheme, at the end of the content delivery phase, each server has received $\ell$ coded chunks for each of its requests. Then, in order to recover the requested file, it has to calculate the inverse of an $\ell \times \ell = \log n \times \log n$ matrix over a finite field $\mathbb{F}_q$. Basically, this introduces a complexity of order $O(\ell^{2.37})$ finite field operations \cite{Coppersmith_1990}. This may not be computationally feasible in certain practical scenarios. Thus, one may ask what are other approaches which achieve the same performance with less computational complexity.

An alternative approach is to use \emph{Fountain-like codes}, originally proposed for packet erasure channels \cite{Byers_1998_Fountain_Code}. This coding technique benefits from a non-uniform coding operator (similar to \eqref{eq:CodingOperator} but with non-uniform distribution over the coefficients $\alpha_r$). The main idea behind these codes is that by optimizing over the coding coefficients distribution, one can design the parity check matrix of these codes such that the complexity of the encoding and decoding algorithms will be significantly reduced. More specifically, in the cache content placement phase in Algorithm \ref{alg:Coded_Cache_Placement}, one can use the encoder of a Raptor code \cite{Shokrollahi06_Raptor_Code}, instead of a uniform operator $\mc{L}$, defined in \eqref{eq:CodingOperator}. It is shown that the Raptor codes have linear encoding and decoding complexity in the codewords length \cite{Shokrollahi06_Raptor_Code}. Translating to our problem setting, this leads to the encoding and decoding complexity of order $\ell=\log(n)$ finite field operations, compared to the aforementioned matrix inversion.

\bibliographystyle{IEEEtran}
\bibliography{AFMP}

\end{document}